\documentclass[final,5p,times,twocolumn]{elsarticle}

\usepackage{amsthm,amsmath,amssymb,mathrsfs,amsfonts,graphicx,graphics,latexsym,exscale,cmmib57,dsfont,bbm,amscd}
\usepackage{color,soul}
\usepackage[colorlinks=true]{hyperref}

\newcommand{\brho}{\boldsymbol{\rho}}

\newcommand{\bA}{\boldsymbol{A}}

\newcommand{\lam}[1]{\lambda_}

\newcommand{\bea}{\begin{eqnarray}}
\newcommand{\eea}{\end{eqnarray}}
\newcommand{\bean}{\begin{eqnarray*}}
\newcommand{\eean}{\end{eqnarray*}}
\newtheorem{prob}{Problem}

\newtheorem*{thm-A}{Theorem A}
\newtheorem*{thm-B}{Theorem B}
\newtheorem*{thm-C}{Theorem C}
\newtheorem*{thm-D}{Theorem D}
\newtheorem*{thm-E}{Theorem E}
\newtheorem*{thm-F}{Theorem F}
\newtheorem*{thm-G}{Theorem G}
\newtheorem*{thm-G'}{Theorem G$^\prime$}
\newtheorem*{thm-H}{Theorem H}

\newtheorem{theorem}{Theorem}

\newtheorem{cor}[theorem]{Corollary}

\theoremstyle{definition}

\newtheorem{example}[theorem]{Example}

\newtheorem{remark}[theorem]{Remark}

\numberwithin{equation}{section}

\begin{document}

\begin{frontmatter}

\title{A criterion of simultaneously symmetrization and spectral finiteness for a finite set of real $2\times 2$ matrices\tnoteref{label1}}

\tnotetext[label1]{Project was supported partly by National Natural Science Foundation of China (No. 11071112) and PAPD of Jiangsu Higher Education Institutions.}%

\author{Xiongping Dai}
\ead{xpdai@nju.edu.cn}
\address{Department of Mathematics, Nanjing University, Nanjing 210093, People's Republic of China}

\begin{abstract}
For $K\ge1$, let there be given an arbitrary finite set $\bA$ consisting of real $2$-by-$2$ matrices
\bean
A_0=\left[\begin{matrix}a&b\\c&d\end{matrix}\right],A_1=\left[\begin{matrix}a_1&r_1b\\r_1c&d_1\end{matrix}\right],\dotsc,A_K=\left[\begin{matrix}a_K&r_Kb\\r_Kc&d_K\end{matrix}\right], \eean
and by $\rho(M)$ it stands for the spectral radius of a square matrix $M$. In this paper, we first show that if $bc>0$ then $\bA$ may be simultaneously symmetrized.
This then implies that if $bc\ge0$,
\bean
\max\{\rho(A_0),\rho(A_1), \dotsc,\rho(A_K)\}=\sup_{n\ge1}\max_{M\in\bA^n}\sqrt[n]{\rho(M)};
\eean
that is, $\bA$ has the spectral finiteness property and then the stability of the switched system defined by $\bA$ is decidable.
\end{abstract}


\end{frontmatter}

\section{Introduction}\label{sec1}%
In this paper, we study the simultaneously symmetrization and then the finite-step realizability of the generalized/joint spectral radius for a finite set of real $2\times 2$ matrices.

\subsection{Criterion of simultaneously symmetrization}\label{sec1.1}
For a real $d\times d$ matrix $A=[a_{ij}]_{1\le i,j\le d}$, it is said to be \textit{symmetric} if $a_{ij}=a_{ji}$ for all $1\le i,j\le d$. A symmetric matrix has many good property like diagonalization. So, symmetrization of matrices is very important for problems involving numerical computation of matrices. In this short paper, we first show a simultaneously symmetrization for a family of real $2\times 2$ matrices, which may be stated as follows:

\begin{theorem}\label{thm1}
Let there be arbitrarily given $K+1$ real $2\times 2$ matrices
\bean
A_0=\left[\begin{matrix}a&b\\c&d\end{matrix}\right],A_1=\left[\begin{matrix}a_1&r_1b\\r_1c&d_1\end{matrix}\right],\dotsc,A_K=\left[\begin{matrix}a_K&r_Kb\\r_Kc&d_K\end{matrix}\right]
\eean
where $K\ge1$. If $bc>0$, then one can find a nonsingular matrix $Q\in\mathds{R}^{2\times 2}$ such that $QA_kQ^{-1}, 0\le k\le K$, all are symmetric.
\end{theorem}

This provides a criterion of simultaneously symmetrizing a finite set of real $2\times 2$ matrices.

\begin{remark}\label{rem2}
In fact, our condition ``$bc>0$" is already very close to ``necessary'', as shown by the following counterexample.
Let
\bean
A_0=\left[\begin{matrix} -3&3.5\\-4&4.5
\end{matrix}\right],\quad
A_1=\left[\begin{matrix}0.5&0\\0&1\end{matrix}\right],\textrm{ where }bc=-14<0.
\eean
Although $A_0$ may be diagonalized and $A_1$ is already diagonal, yet it will be proved in Section~\ref{sec2} that $\{A_0,A_1\}$ cannot be simultaneously symmetrized.
\end{remark}

As an application, we will see that Theorem~\ref{thm1} is important for the numerical computation of the generalized spectral radius of a family of real $2\times 2$ matrices.

\subsection{Spectral finiteness for a finite set of real $2\times 2$ matrices}\label{sec1.2}

Throughout this paper, $\rho(M)$ will stand for the usual spectral radius of a square matrix $M$. For an arbitrary family of real matrices
\bean
\bA=\{A_0,\dotsc,A_K\}\subset\mathds{R}^{d\times d}
\eean
where $2\le d<+\infty$, its \textit{generalized spectral radius}, first introduced by Daubechies and Lagarias in \cite{DL92-01}, is defined by
\bean
\pmb{\rho}=\sup_{n\ge1}\max_{M\in\bA^n}\sqrt[n]{\rho(M)}\quad\left(\,=\limsup_{n\to+\infty}\max_{M\in\bA^n}\sqrt[n]{\rho(M)}\right),
\eean
where
\bean
\bA^n=\{M_1\dotsm M_n\colon M_i\in\bA\textrm{ for }1\le i\le n\}\quad\forall n\ge1.
\eean
According to the Berger-Wang spectral formula~\cite{BW92}, this quantity is very important for many pure and applied mathematics branches like numerical computation of matrices, differential equations, coding theory, wavelets, stability analysis of random matrix, control theory, combinatorics, and so on. See, for example, \cite{DL92-01, Gur95}.

Therefore, the following finite-step realization question for the accurate computation of the spectral radius $\pmb{\rho}$ becomes very interesting and important.

\begin{prob}\label{prob1}
Does there exist a finite-length word which realize $\pmb{\rho}$ for $\bA$; i.e.,
\bean
\pmb{\rho}=\max_{n\ge1}\max_{M\in\bA^n}\sqrt[n]{\rho(M)}\,\textrm{?}
\eean
In other words, does there exist any $M\in\bA^n$ for some $n\ge1$ such that 
\bean
\pmb{\rho}=\sqrt[n]{\rho(M)}\,\textrm{?}
\eean
\end{prob}

If one can find some word, say $M\in\bA^n$, for some $n\ge1$, such that $\pmb{\rho}=\sqrt[n]{M}$, then $\bA$ is said to possess \textit{the spectral finiteness property}.

This problem is equivalent to the following stability question:

\begin{prob}\label{prob2}
If the periodic stability $($i.e. $\rho(M)<1$ for any finite-length words $M\in\bigcup_{n\ge1}\bA^n$$)$ is satisfied then, does it hold the absolute stability:
\bean
\max_{M\in\bA^n}\|M\|\to0\quad \textrm{as }n\to+\infty\textrm{?}
\eean
\end{prob}

This spectral finiteness property, or equivalently, ``periodic stability $\Rightarrow$ absolute stability'', of $\bA$ was conjectured, respectively, by Pyatnitski\v{i} (see e.g.~\cite{PR91,SWMW07}), Daubechies and Lagarias in~\cite{DL92-01}, Gurvits in~\cite{Gur95}, and by Lagarias and Wang in~\cite{LW95}. It has been disproved first by Bousch and Mairesse in \cite{BM}, and then by Blondel \textit{et al.} in \cite{BTV}, by Kozyakin in~\cite{Koz05, Koz07}, all offered the existence of counterexamples in the case where $d=2$; moreover, an explicit expression for such a counterexample has been found in the recent work of Hare \textit{et al.}~\cite{HMST}.

However, an affirmative solution to Problem~\ref{prob1} (or equivalently, to Problem~\ref{prob2}) is very important; this is because it implies an effective computation of $\pmb{\rho}$ and decidability of stability of $\bA$ by only finitely many steps of computations. There have been some sufficient (and necessary) conditions for the spectral finiteness property for some systems $\bA$, based on and involving Barabanov norms, polytope norms, ergodic theory or some limit properties of $\bA$, for example, in Gurvits~\cite{Gur95}, Lagarias and Wang~\cite{LW95}, Guglielmi, Wirth and Zennaro~\cite{GWZ05}, Kozyakin~\cite{Koz07}, Dai, Huang and Xiao~\cite{DHX-pro}, and Dai and Kozyakin~\cite{DK}. But these theoretic criteria seems to be difficult to be directly employed to judge whether or not an explicit family $\bA$ or even a pair $\{A,B\}\subset\mathds{R}^{2\times 2}$ have the spectral finiteness property.

From literature, as far we know, there are only few results on such an explicit family of matrices $\bA$.

\begin{thm-A}[{Theys~\cite{Theys}, also see \cite[Proposition~4]{JB08}}]
If $A_0,\dotsc,A_K\in\mathds{R}^{d\times d}$ all are symmetric matrices, then the spectral finiteness property holds for $\bA$. In fact, there holds
\bean
\brho=\max_{0\le k\le K}\rho(A_k).
\eean
\end{thm-A}

For any matrix $A$, by $A^T$ it denotes the transpose of $A$. An generalization of Theorem~A is the following

\begin{thm-B}[{Plischke and Wirth~\cite[Proposition~18]{PW:LAA08}}]
If the system $\bA=\{A_0,\dotsc,A_K\}\in\mathds{R}^{d\times d}$ is symmetric, i.e. $A_k^T\in\bA$ for all $0\le k\le K$, then the spectral finiteness property holds for $\bA$.
\end{thm-B}

For a pair of matrices, there are the following results.

\begin{thm-C}[Jungers and Blondel~\cite{JB08}]
If $A_0, A_1$ are $2\times 2$ binary matrices, i.e. $A_0,A_1\in\{0,1\}^{2\times 2}$, then the spectral finiteness property holds for $\{A_0,A_1\}$.
\end{thm-C}

A more general result than the statement of Theorem~C is the following

\begin{thm-D}[Cicone \textit{et al.}~\cite{CGSZ10}]
If $A_0, A_1$ are $2\times 2$ sign-matrices, that is, $A_0,A_1$ belong to $\{0,\pm1\}^{2\times 2}$, then the spectral finiteness property holds for $\{A_0,A_1\}$.
\end{thm-D}

The followings are other different type of results.

\begin{thm-E}[Dai \textit{et al.}~\cite{DHLX11}]
If one of $A, B\in\mathds{R}^{d\times d}$ is of rank one, then there holds the spectral finiteness property for $\{A,B\}$.
\end{thm-E}

\begin{thm-F}[Dai, Huang and Xiao~\cite{DHX-siam}]
If, for $A, B\in\mathds{R}^{d\times d}$, there is a symmetric positive-definite matrix $P$ such that
\bean
P-A^TPA\ge0\quad\textrm{and}\quad P-B^TPB\ge0,
\eean
then the spectral finiteness property holds for $\{A,B\}$ in the case $2\le d\le3$.
\end{thm-F}

Using our symmetrization Theorem~\ref{thm1}, we can prove the following finiteness result:

\begin{theorem}\label{thm3}
Let there be arbitrarily given $K+1$ real $2\times 2$ matrices
\bean
A_0=\left[\begin{matrix}a&b\\c&d\end{matrix}\right],A_1=\left[\begin{matrix}a_1&r_1b\\r_1c&d_1\end{matrix}\right],\dotsc,A_K=\left[\begin{matrix}a_K&r_Kb\\r_Kc&d_K\end{matrix}\right], \eean
where $K\ge1$.
If $bc\ge0$ then $\bA=\{A_0,\dotsc,A_K\}$ has the spectral finiteness property and moreover
\bean
\brho=\max_{0\le k\le K}\rho(A_k).
\eean
\end{theorem}

\begin{proof}
If $bc=0$ then the statement holds trivially. Now let $bc>0$. From Theorem~\ref{thm1}, one can find some nonsingular matrix $Q$ such that $QA_kQ^{-1}$, $0\le k\le K$, all are symmetric. Then, the statement of Theorem~\ref{thm3} follows immediately from Theorem~A, also from Theorem~B.
\end{proof}

As a result of Theorem~\ref{thm3}, we can obtain the following

\begin{cor}\label{cor4}
Let $A,B\in\mathds{R}^{2\times 2}$ be a pair of matrices such that
\begin{equation*}
A=\left[\begin{matrix}\lambda_1&0\\0&\lambda_2\end{matrix}\right],\quad B=\left[\begin{matrix}a&b\\c&d\end{matrix}\right].
\end{equation*}
If $bc\ge0$ then there holds the spectral finiteness property for $\{A,B\}$. More precisely, if $bc\ge0$ then $\brho=\max\{\rho(A), \rho(B)\}$.
\end{cor}

Without the constraint condition $bc\ge0$ in Corollary~\ref{cor4}, a special case might be simply observed.

\begin{theorem}
Let $A,B\in\mathds{R}^{2\times 2}$ be a pair of matrices such that $A=\mathrm{diag}(\lambda_1,\lambda_2)$ and $B=\left[\begin{matrix}0&b\\c&0\end{matrix}\right]$. Then $\{A,B\}$ has the spectral finiteness property with $\brho=\max\{\rho(A), \rho(B)\}$.
\end{theorem}

\begin{proof}
Let $\rho(A)=\max\{|\lambda_1|,|\lambda_2|\}<1$ and $\rho(B)=\sqrt{|bc|}<1$. Let $\{(m_k,n_k)\}_{k=1}^{+\infty}$ be an arbitrary sequence of positive integer pairs. We claim that
\bean
\|A^{m_1}B^{n_1}A^{m_2}B^{n_2}\dotsm A^{m_k}B^{n_k}\|_2\to0\quad \textrm{as }k\to+\infty,
\eean
where $\|\cdot\|_2$ denotes the matrix norm induced by the standard Euclidean vector norm on $\mathds{R}^2$. In fact, the claim follows from
\bean
A^m=\left[\begin{matrix}\lambda_1^m&0\\0&\lambda_2^m\end{matrix}\right]\quad \textrm{and}\quad B^n=\begin{cases}(bc)^{n^\prime}I_2& \textrm{if }n=2n^\prime,\\(bc)^{n^\prime}B& \textrm{if }n=2n^\prime+1.\end{cases}
\eean
Then, this claim implies that $\brho=\max\{\rho(A), \rho(B)\}$.
\end{proof}

\subsection{Outline}
This paper is simply organized as follows. We will prove Theorem~\ref{thm1} and Remark~\ref{rem2} in Section~\ref{sec2}. Finally, we will end this paper with some examples in Section~\ref{sec3}.

\section{Simultaneously symmetrization}\label{sec2}

This section is mainly devoted to proving our criterion of simultaneously symmetrizing, i.e., Theorem~\ref{thm1}.

\begin{proof}[Proof of Theorem~\ref{thm1}]
Let there be arbitrarily given $K+1$ real $2\times 2$ matrices
\bean
A_0=\left[\begin{matrix}a&b\\c&d\end{matrix}\right],A_1=\left[\begin{matrix}a_1&r_1b\\r_1c&d_1\end{matrix}\right],\dotsc,A_K=\left[\begin{matrix}a_K&r_Kb\\r_Kc&d_K\end{matrix}\right], \eean
where $K\ge1$, such that $bc>0$. Let
\bean
Q=\left[\begin{matrix}q_1&0\\0&q_2\end{matrix}\right]\quad \textrm{such that }\; q_1q_2\not=0\textrm{ and }\frac{q_1}{q_2}=\sqrt{\frac{c}{b}}.
\eean
Then,
\bean
QA_0Q^{-1}&=&\left[\begin{matrix}a&\sqrt{bc}\\\sqrt{bc}&d\end{matrix}\right],\\ QA_1Q^{-1}&=&\left[\begin{matrix}a_1&r_1\sqrt{bc}\\r_1\sqrt{bc}&d_1\end{matrix}\right],\\
\vdots&\dotsm&\vdots,\\
QA_KQ^{-1}&=&\left[\begin{matrix}a_K&r_K\sqrt{bc}\\r_K\sqrt{bc}&d_K\end{matrix}\right],
\eean
they are symmetric. This proves Theorem~\ref{thm1}.
\end{proof}

We now turn to the proof of Remark~\ref{rem2}.

Let
\bean
A_0=\left[\begin{matrix} -3&3.5\\-4&4.5
\end{matrix}\right],\;
A_1=\left[\begin{matrix}0.5&0\\0&1\end{matrix}\right],\;\textrm{ where }bc=-14<0,
\eean
as in Remark~\ref{rem2}. Put
\bean
Q=\left[\begin{matrix}-0.5& 1 \\  0 & 1\end{matrix}\right],
\eean
then we have
\bean
Q^{-1}=\left[\begin{matrix}-2& 2\\ 0 & 1\end{matrix}\right].
\eean
So,
\bean
B_0:=Q^{-1}A_0Q=\left[\begin{matrix}1 & 0 \\  2 & 0.5
\end{matrix}\right]
\eean
and
\bean
B_{1}:=Q^{-1}A_1Q=\left[\begin{matrix}0.5 & 1 \\  0 & 1\end{matrix}\right].
\eean

According to Kozyakin~\cite[Theorem~10, Lemma~12 and Theorem~6]{Koz07}, there follows that: \textit{There always exists a pair of real numbers $\alpha>0,\beta>0$ such that $\{\alpha B_0,\beta B_1\}$ does not have the spectral finiteness property}.

Thus, if $\{A_0,A_1\}$ might be simultaneously symmetrized, then $\{\alpha A_0,\beta A_1\}$ and hence $\{\alpha B_0,\beta B_1\}$ have the spectral finiteness property from Theorem~\ref{thm3}, for all $\alpha>0,\beta>0$. This is a contradiction. Therefore, $\{A_0,A_1\}$ cannot be simultaneously symmetrized.

This proves the statement of Remark~\ref{rem2}.
Meanwhile this argument shows that the constraint condition ``$bc\ge0$" in Theorem~\ref{thm3} and even in Corollary~\ref{cor4} is crucial for the spectral finiteness property in our situation.

\medskip
Given an arbitrary set $\bA=\{A_0,\dotsc,A_K\}\subset\mathds{R}^{d\times d}$, although its periodic stability implies that it is stable almost surely in terms of arbitrary Markovian measures as shown in Dai, Huang and Xiao~\cite{DHX11-aut} for the discrete-time case and in Dai~\cite{Dai-JDE} for the continuous-time case, yet its absolute stability is
generally undecidable; see, e.g., Blondel and Tsitsiklis~\cite{BT97, BT00, BT00-aut}.

However, Theorem~\ref{thm3} proved in Section~\ref{sec1.2} is equivalent to the statement\,---\,``periodic stability $\Rightarrow$ absolute stability'', i.e., Problem~\ref{prob2}, under suitable additional conditions.

\begin{theorem}\label{thm6}
Let $\bA=\{A_0,\dotsc,A_K\}\subset\mathds{R}^{2\times 2}$ be such that
\bean
A_0=\left[\begin{matrix}a&b\\c&d\end{matrix}\right],A_1=\left[\begin{matrix}a_1&r_1b\\r_1c&d_1\end{matrix}\right],\dotsc,A_K=\left[\begin{matrix}a_K&r_Kb\\r_Kc&d_K\end{matrix}\right], \eean
where $K\ge1$ and $bc\ge0$.
Then $\bA$ is absolutely stable if and only if $\rho(A_k)<1$ for all $0\le k\le K$.
\end{theorem}

\begin{proof}
This statement comes immediately from Theorem~\ref{thm3}. In fact, Theorem~\ref{thm3} implies $\brho<1$ $\mathit{iff}$ $\rho(A_k)<1$ for all $0\le k\le K$ and hence $\bA$ is absolutely stable $\mathit{iff}$ $\rho(A_k)<1$ for all $0\le k\le K$; see, e.g., \cite{Bar, Gur95, SWP}.
\end{proof}

This shows that the absolute stability of the switched system induced by $\bA$ is decidable in the situation of Theorem~\ref{thm6}.

\section{Examples of stability}\label{sec3}
In this section, we consider some explicit examples using Theorem~\ref{thm3} and Corollary~\ref{cor4}.

\subsection{Applications of Corollary~\ref{cor4}}%
For any two real $2\times 2$ matrices $A,B$, to utilize our Corollary~\ref{cor4}, the first step is to diagonalize one of $A,B$. So, we need the Diagonalization Theorem: \textit{An $n\times n$ matrix $A$ is diagonalizable if and only if $A$ has $n$ linearly independent eigenvectors}.

\begin{example}\label{example7}
Let $A=\left[\begin{matrix}2&1\\0&1\end{matrix}\right]$ and $B=\left[\begin{matrix}-\frac{5}{2}&\frac{2\sqrt{3}-11}{2}\\1&4\end{matrix}\right]$. We assert that $\{A,B\}$ has the spectral finiteness property.

In fact, since
\bean
A_1:=\left[\begin{matrix}1&1\\0&1\end{matrix}\right]\left[\begin{matrix}2&1\\0&1\end{matrix}\right]\left[\begin{matrix}1&-1\\0&1\end{matrix}\right]=\left[\begin{matrix}2&0\\0&1\end{matrix}\right]
\eean
and
\bean
A_0:=\left[\begin{matrix}1&1\\0&1\end{matrix}\right]\left[\begin{matrix}-\frac{5}{2}&\frac{2\sqrt{3}-11}{2}\\1&4\end{matrix}\right]\left[\begin{matrix}1&-1\\0&1\end{matrix}\right]=\left[\begin{matrix}-\frac{3}{2}&\sqrt{3}\\1&3\end{matrix}\right],
\eean
it follows, from Corollary~\ref{cor4}, that $\{A,B\}$ has the spectral finiteness property with
\bean
\brho=\rho(B)=\frac{1}{2}\left(3+\sqrt{27+4\sqrt{3}}\right).
\eean
\end{example}

\begin{example}\label{example8}
Let $A=\left[\begin{matrix}0.95&0.03\\0.05&0.97\end{matrix}\right]$ and $B=\left[\begin{matrix}0&b\\c&0\end{matrix}\right]$, where $b,c\in\mathds{R}$ such that \begin{equation*}
225b^2-34bc+c^2\le0.
\end{equation*}
We now consider the spectral finiteness property of $\{A,B\}$.

The eigenvalues of $A$ are $1$ and $0.92$, their corresponding eigenvectors are respectively $(3,5)^T$ and $(1,-1)^T$. We put
\bean
P=\left[\begin{matrix}3&1\\5&-1\end{matrix}\right].
\eean
Then
\bean
P^{-1}=\left[\begin{matrix}1/8&1/8\\5/8&-3/8\end{matrix}\right]\quad \textrm{and}\quad P^{-1}AP=\left[\begin{matrix}1&0\\0&0.92\end{matrix}\right].
\eean
Since
\bean
P^{-1}BP=\frac{1}{8}\left[\begin{matrix}*&c-9b\\-c+25b&\star\end{matrix}\right],
\eean
there follows $(c-9b)(-c+25b)\ge0$. So, $\{A,B\}$ has the spectral finiteness property from Corollary~\ref{cor4} such that
\begin{equation*}
\brho=\max\left\{1, \sqrt{|bc|}\right\}.
\end{equation*}
\end{example}

\begin{example}\label{example9}
Let $A_0=\left[\begin{matrix}a&b\\0&1\end{matrix}\right]$ and $A_1=\left[\begin{matrix}1&0\\c&d\end{matrix}\right]$, where the constants $a,b,c,d\in\mathds{R}$ with $a\not=1$.

If $ad=0$ then either $\mathrm{rank}(A_0)=1$ or $\mathrm{rank}(A_1)=1$ and so $\{A_0,A_1\}$ has the spectral finiteness property from Theorem~E.

If $bc=0$ then either $b=0$ or $c=0$. So $\{A_0,A_1\}$ has the spectral finiteness property from Corollary~\ref{cor4}.

Next, we let $bc\not=0$ and define
\bean
Q=\left[\begin{matrix}\frac{a-1}{b}&1\\0&1\end{matrix}\right].
\eean
Then,
\bean
Q^{-1}=\left[\begin{matrix}\frac{b}{a-1}&-\frac{b}{a-1}\\0&1\end{matrix}\right]
\eean
and
\bean
QA_0A^{-1}=\left[\begin{matrix}a&0\\0&1\end{matrix}\right],\quad QA_1A^{-1}=\left[\begin{matrix}1+\frac{bc}{a-1}&\frac{(d-1)(a-1)-bc}{a-1}\\\frac{bc}{a-1}&d-\frac{bc}{a-1}\end{matrix}\right].
\eean
Note that
\bean
\frac{(d-1)(a-1)-bc}{a-1}\times\frac{bc}{a-1}\ge0
\eean
if and only if
\bean
[(1-a)(1-d)-bc]\times bc\ge0.
\eean
Hence, if either $a\not=1$ or $d\not=1$ and $[(1-a)(1-d)-bc]\times bc\ge0$, then $\{A_0,A_1\}$ has the spectral finiteness property.

If $a=d=1$ and $bc\ge1$, then $\{A_0,A_1\}$ has the spectral finiteness property from Kozyakin~\cite[Theorem~10, Lemma~12 and Theorem~6]{Koz07}.
\end{example}

\subsection{Applications of Theorem~\ref{thm3}}%
Applying Theorem~\ref{thm3}, we consider the following

\begin{example}\label{example10}
Let
\bean
A_0=\left[\begin{matrix}\sqrt{3}&1\\2&1.3\end{matrix}\right], A_1=\left[\begin{matrix}\sqrt{2}&10\\20&\sqrt{7}\end{matrix}\right],A_2=\left[\begin{matrix}-1&0.1\\0.2&\sqrt{5}\end{matrix}\right].
\eean
Then from Theorem~\ref{thm3}, if follows that $\{A_0,A_1,A_2\}$ has the spectral finiteness property.
\end{example}
\bigskip
\subsection*{\textbf{Acknowledgments}}%
The author would like to thank professors Y.~Huang and M.~Xiao for some helpful discussion, and particularly, he is grateful to professor Victor Kozyakin for some useful comments.
\bibliographystyle{amsplain}

\end{document}